\documentclass[conference]{IEEEtran}
\pdfoutput=1 

\usepackage[margin=0.75in]{geometry}
\usepackage{amssymb,epsf,amsmath,amsthm,verbatim}
\allowdisplaybreaks[4]
\usepackage{wasysym}
\usepackage{graphicx}

\usepackage{mathtools}
\usepackage{dsfont}
\usepackage{verbatim}
\usepackage{array}   
\newcolumntype{L}{>{$}l<{$}}

\usepackage[inline,shortlabels]{enumitem}
\setlist[itemize]{leftmargin=*}
\setlist[enumerate]{leftmargin=*}

\DeclareMathOperator*{\argmax}{argmax}
\newcommand{\mA}{\mathcal{A}}
\newcommand{\mX}{\mathcal{X}}

\newcommand{\mN}{\mathcal{N}}
\newcommand{\upi}{\underline{\pi}}
\renewcommand{\hat}{\widehat}

\newcommand{\mP}{\mathbb{P}}
\newcommand{\mE}{\mathbb{E}}

\newtheorem{theorem}{Theorem}
\newtheorem{lemma}[theorem]{Lemma}

\newtheorem{corollary}[theorem]{Corollary}

\begin{document}

\newgeometry{tmargin=1in,lmargin=0.75in,rmargin=0.75in,bmargin=0.75in,}

\title{Structured Perfect Bayesian Equilibrium in Infinite Horizon Dynamic Games with Asymmetric Information}
\author{Abhinav Sinha and Achilleas Anastasopoulos}
\maketitle


\begin{abstract}
	In dynamic games with asymmetric information structure, the widely used concept of equilibrium is perfect Bayesian equilibrium (PBE). This is expressed as a strategy and belief pair that simultaneously satisfy sequential rationality and belief consistency.
	Unlike symmetric information dynamic games, where subgame perfect equilibrium (SPE) is the natural equilibrium concept, to date there does not exist a universal algorithm that decouples the interdependence of strategies and beliefs over time in calculating PBE.
	In this paper we find a subset of PBE for an infinite horizon discounted reward asymmetric information dynamic game. We refer to it as Structured PBE or SPBE; in SPBE, any agents' strategy depends on the public history only through a common public belief and on private history only through the respective agents' latest private information (his private type).
	The public belief acts as a summary of all the relevant past information and it's dimension does not increase with time. The motivation for this comes the common information approach proposed in Nayyar et al. (2013) for solving decentralized team (non-strategic) resource allocation problems with asymmetric information. 
	We calculate SPBE by solving a single-shot fixed-point equation and a corresponding forward recursive algorithm. We demonstrate our methodology by means of a public goods example.
	%
\end{abstract}

\section{Introduction}

Dynamic games with symmetric information among agents has been studied well in game theory literature~\cite{fudenbergtirole,basar1999,samuelson2006,krivan2014}. The appropriate equilibrium concept is subgame perfect equilibrium (SPE) which is a refinement of Nash equilibrium. A strategy profile is SPE if its restriction to any subgame of the original game is also an SPE of the subgame. However there are models of interest where strategic agents interact repeatedly whilst observing private signals. For instance, relay scheduling with private queue length information at different relays~\cite{vasal14}. Another example is that of Bayesian learning games~\cite{bikhchandani1992,acemoglu2011,vijay2014,vasal16} such as when customers post reviews on websites such as Amazon; clearly agents posses their own experience of the object they bought and this constitutes private information. These examples lead to a model with asymmetric information, where the notion of SPE is ineffective.

Appropriate equilibrium concepts in such a case consist of strategy profiles and beliefs. These include weak perfect Bayesian equilibrium, perfect Bayesian equilibrium (PBE) and Sequential equilibrium. PBE is the most commonly used notion. The requirement of PBE is sequential rationality of strategy as well as belief consistency. 
Due to the fact that future beliefs depend on past strategy and strategies are sequentially rational based on specific beliefs, finding a pair of strategy and belief that satisfy PBE requirement is a difficult problem and usually reduces to a large fixed-point equation over the entire time horizon.
To date, there is no universal algorithm that provides simplification by decomposing the aforementioned fixed-point equation for calculating PBEs.

Our motivation stems from the work of Nayyar et al.~\cite{nayyar2013}. Authors in~\cite{nayyar2013} consider a decentralized dynamic team problem (non-strategic agents) with asymmetric information. Each agents observes a private type which evolves with time as a controlled Markov process. They introduce a common information based approach, whereby each agent calculates a belief on every agents' current private type. This belief is defined in such a manner that despite having asymmetric information, agents can agree on it. They proceed to show optimality of policies that depend on private history only through this belief and respective agents' current private type.  

The common information based approach has been studied for finite horizon dynamic games with asymmetric information in~\cite{gupta2014,nayyar2014,vasal2015aa,vasal15b,vasal16}. Of these,~\cite{gupta2014,nayyar2014} consider models where the aforementioned belief from the common information approach is updated independent of strategy.
This implies that the simultaneous requirements of sequential rationality and belief consistency can be decoupled, resulting in a simplification in calculating PBEs. However these models produce non-signaling PBEs. For the models studied  in~\cite{vasal2015aa,vasal15b,vasal16} belief updates depend on strategies. Consequently the resulting PBEs are signaling equilibria. The problem of finding them though, becomes more complicated.

Other than the common information based approach, Li et al.~\cite{shamma14} consider a finite horizon zero-sum dynamic game, where at each time only one agent out of the two knows the state of the system. The value of the game is calculated by formulating an appropriate linear program. 

All works listed above for dynamic games consider a finite horizon. In this paper, we deal with an infinite horizon discounted reward dynamic game.
Cole et al.~\cite{cole2001} consider an infinite horizon discounted reward dynamic game where actions are only privately observable. They provide a fixed-point equation for calculating a subset of sequential equilibrium, which is referred to as Markov private equilibrium (MPrE). In MPrE strategies depend on history only through the latest private observation. 

\restoregeometry

\subsection*{Contributions}

In this work we follow the models of~\cite{vasal2015aa,vasal15b,vasal16} where beliefs are updated dependent on strategy and thus signaling PBEs are considered.

This paper provides a one-shot fixed-point equation and a forward recursive algorithm that together generate a subset of perfect Bayesian equilibria for the infinite horizon discounted reward dynamic game with asymmetric information. The common information based approach of~\cite{nayyar2013} is used and the strategies are such that they depend on the public history only through the common belief (which summarizes the relevant information) and on private history only through respective agents' current private type. The PBE thus generated are referred to as Structured PBE or SPBE and posses the property that the mapping between belief and strategy is stationary (w.r.t. time).

The provided methodology (consisting of solving a one-shot fixed-point equation) provides a decomposition of the interdependence between belief and strategy in PBEs and enables a systematic evaluation of SPBEs.

The model allows for signaling amongst agents as beliefs depend on strategies. This is because the source of asymmetry of information in our model is that each agent has a privately observed type that affects the utility function of all the agents. These types evolve as a controlled Markov process and depend on actions of all agents (the actions are assumed to be common knowledge). 

Our methodology for proving our results is as follows: first we extend the finite horizon model of~\cite{vasal2015aa} to include belief-based terminal reward. This does not change the proofs significantly but helps us simplify the exposition and proofs of the infinite horizon case. In the next step, the infinite horizon results are obtained with the help of the finite horizon ones through continuity arguments as horizon $ T \rightarrow \infty $.

We demonstrate our methodology of finding SPBEs through a concrete public goods game.


The remainder of this paper is organized as follows: Section~\ref{secmodel} introduces the model for the dynamic game with private types and discounted reward. We consider two versions of the problem, finite and infinite horizon. Section~\ref{secfh} defines the backward recursive algorithm for calculating SPBEs in the finite horizon model. The results in this section are an adaptation of the finite horizon results from~\cite{vasal2015aa}. This is done so that the finite horizon results can be appropriately used in Section~\ref{secih}.  Section~\ref{secih} contains the central result of this paper, namely the fixed-point equation that defines the set of SPBE strategy for the infinite horizon game. Finally, Section~\ref{secexample} discusses  a concrete example of a public goods game with two agents. This example is an infinite horizon version of the finite horizon public goods example from~\cite[ch.~8, Example~8.3]{fudenbergtirole}.

\section{Model for Dynamic Game} \label{secmodel}

We consider a dynamical system with $ N $ strategic agents, denoted by $ \mN $. Associated with each agent $ i $ at any time $ t \ge 1 $, is a type $ x_t^i \in \mX^i $. The set of types is denoted by $ \mX^i $ and is assumed to be finite. We assume that each agent $ i $ can observe their own type $ x_t^i $ but not that of others. Denote the profile of types at time $ t $ by $ x_t = (x_t^i)_{i\in \mN} $.

Each agent $ i $ at any time $ t \ge 1 $, after observing their private type $ x_t^i $ takes action $ a_t^i \in \mA^i $. The set of available actions to agent $ i $ is denoted by $ \mA^i $ and is assumed to be finite. It is assumed that the action profile is publicly observed and is common knowledge.

The type of each agent evolves over time as a controlled Markov process independent of other agents given the action profile $ a_t = (a_t^i)_{i\in \mN} $. We assume a time-homogeneous kernel $ Q^i $ for evolution of agent $ i $'s type i.e., $ \mP(x_{t+1} \mid x_{1:t},a_{1:t}) = \mP(x_{t+1} \mid x_{t},a_{t}) = \prod_{i=1}^N Q^i(x_{t+1}^i \mid x_t^i,a_t) $. The notation we use is $ a_{1:t} = (a_k)_{k=1,\ldots,t} $. Initial belief, at time $ 1 $, is assumed as $ \mP(x_1) = \prod_{i=1}^N Q_0^i(x_1^i) $.

Associated with each agent $ i $ is a reward function $ R^i $ that depends on the type and action profiles i.e., reward at time $ t $ for agent $ i $ is $ R^i(x_t,a_t) $. Furthermore, rewards are accumulated over time in a discounted manner with a common discount factor $ \delta \in (0,1) $. We consider two versions of the problem - finite horizon $ (T) $ and infinite horizon. For the finite horizon problem we introduce a terminal reward. We first define beliefs $ \upi_t $ since the terminal reward is defined as a function of beliefs.



Define the public history $ h_t^c $ at time $ t $ as the set of publicly observed actions i.e., $ h_t^c = a_{1:t-1} \in (\times_{i=1}^N \mA^i)^{t-1} $. Denote the set of public histories by $ \mathcal{H}_t^c $. Define private history $ h_t^i $ of any agent $ i $ at time $ t $ as the set of privately observed types $ x_{1:t}^i \in (\mX^i)^t $ of agent $ i $ and publicly observed action profiles $ a_{1:t-1} $ i.e., $ h_t^i = (x_{1:t}^i,a_{1:t-1}) $. Denote the set of private histories by $ \mathcal{H}_t^i $. Note that private history $ h_t^i $ contains public history $ h_t^c $. Also denote the overall history at time $ t $ by $ h_t = (x_{1:t},a_{1:t-1}) $ and the set of all such histories by $ \mathcal{H}_t $.

Any strategy used by agent $ i $ can be represented as $ g^i = (g_t^i)_{t \ge 1} $ with $ g_t^i: \mathcal{H}_t^i \rightarrow \Delta(\mA^i) $. The notation we use is: $ \Delta(S) $ is the set of all probability distributions over the elements of the set $ S $.

The belief $ \upi_t \in \times_{j=1}^N \Delta(\mX^j) $, at any time $ t $, depends on the strategy profile $ g_{1:t-1} = (g_{1:t-1}^j)_{j \in \mN} $. Given a strategy profile $ g_{1:t-1} $, we define the belief as follows: $ \upi_t(x_t) = \big( \pi_t^i(x_t^i) \big)_{i\in \mN} $, with $ \pi_t^i(x_t^i) = \mP^g(x_t^i \mid h_t^c) $ if $ \mP^g(h_t^c) > 0 $, else
\begin{gather}\label{eqpidef}
\pi_t^i(x_t^i) = \sum_{x_{t-1}^i \in \mX^i} Q^i(x_t^i \mid x_{t-1}^i,a_{t-1}) \pi_{t-1}^i(x_{t-1}^i).   
\end{gather}
Initial belief is $ \pi_1^i(x_1^i) = Q_0^i(x_1^i) $. The above definition ensures that the belief random variable $ \Pi_t $, at any time $ t $, is measurable, even if the set of histories resulting in any instance $ \pi_t $ has zero measure under strategy $ g $.

Finally, in the finite horizon game, for each agent $ i $ there is a terminal reward $ G^i(\pi_{T+1},x_{T+1}^i) $ that depends on the terminal type of agent $ i $ and the terminal belief. It is assumed that $ G^i(\cdot) $ is absolutely bounded.





\section{Finite Horizon} \label{secfh} 

In this section, we state and prove properties of the finite horizon backwards recursive algorithm for calculating SPBE, adopted from~\cite{vasal2015aa} with minor modifications to accommodate for terminal rewards (specifically belief-based). The results from this section are used in Section~\ref{secih} to prove the infinite horizon equilibrium result.

To distinguish quantities defined in this section with infinite horizon, we add a superscript $ T $ to all the quantities. 

\subsection{Finite Horizon problem and Backwards Recursion} \label{secfhalg}

Consider a finite horizon, $ T > 1 $, problem in the above dynamic game model. 


Define the value functions $ \big(V_t^{i,T}: \times_{j \in \mN} \Delta(\mX^j) \times \mX^i \rightarrow \mathbb{R} \big)_{i \in \mN,t\in \{ 1,\ldots,T \} } $ and strategy $ \big( \tilde{\gamma}_t^{i,T} \big)_{i \in \mN,t\in \{ 1,\ldots,T \} } $ backwards inductively as follows. 

\begin{enumerate}
	\item  $ \forall  $ $ i \in \mN $, set $ V_{T+1}^{i,T} \equiv G^i $.
	
	\item For any $ t \in \{1,\ldots,T\} $ and $ \upi_t $, solve the following fixed-point equation in $ (\tilde{\gamma}_t^{i,T})_{i\in \mN} $. $ \forall $ $ i \in \mN $, $ x^i \in \mX^i $, 
	\begin{multline}\label{eqfhalg}
	\tilde{\gamma}_t^{i,T}(\cdot \mid x_t^i) 
	\in \!\!\! \argmax_{\gamma_t^i(\cdot \mid x_t^i) \in \Delta(\mA^i)} \!\!\! \mE^{\gamma_t^i(\cdot \mid x_T^i),\tilde{\gamma}_t^{-i,T},\upi_t^{-i}} \big[ R^i(X_t,A_t) 
	\\
	+ \delta V_{t+1}^{i,T}\big( \big[F(\pi_t^j, \tilde{\gamma}_t^{j,T}, A_t)\big]_{j=1}^N , X_{t+1}^i \big) \mid \upi_t,x_t^i \big]
	\end{multline} 
	(see below for the various quantities involved in the above expression).
	\item Then define,
	\begin{multline} \label{eqfhalg2}
	V_t^{i,T}(\upi_t,x_t^i) 
	= \mE^{\tilde{\gamma}_t^{i,T}(\cdot \mid x_t^i),\tilde{\gamma}_t^{-i,T},\upi_t^{-i}} \big[ R^i(X_t,A_t) 
	\\
	+ \delta V_{t+1}^{i,T}\big( \big[F(\pi_t^j, \tilde{\gamma}_t^{j,T}, A_t)\big]_{j=1}^N , X_{t+1}^i \big) \mid \upi_t,x_t^i \big]
	\end{multline}
	
\end{enumerate}

\vspace{3ex}


Denote by $ \theta_t^i $ the mapping $ \upi_t \mapsto \tilde{\gamma}_t^{i,T}  $ i.e., $ \tilde{\gamma}_t^{i,T} = \theta_t^i[\upi_t] $. 

In~\eqref{eqfhalg}, the expectation is with the following distribution
\begin{multline}
(X_t^{-i},A_t^i,A_t^{-i},X_{t+1}^{i}) \sim \pi_t^{-i}(x_t^{-i})\gamma_t^{i,T}(a_t^i \mid x_t^i) 
\\
\tilde{\gamma}_t^{-i,T}(a_t^{-i} \mid x_t^{-i}) Q^i(x_{t+1}^{i} \mid x_t^i,a_t).
\end{multline}
and
\begin{gather} 
\nonumber 
F(\pi^j,\gamma^j,a)(x^{\prime j}) 
\\
\label{eqf}
= \left\{ \!\!\!\!
\begin{array}{ll}
\displaystyle\frac{\sum_{x^j \in \mX^j} \pi^j(x^j) \gamma^j(a^j \mid x^j) Q^j(x^{\prime j} \mid x^j,a)}{\sum_{\tilde{x}^j} \pi^j(\tilde{x}^j) \gamma^j(a^j \mid \tilde{x}^i)}  \!\!\! & \mbox{if } \text{Den.} > 0 
\\[2ex]
\sum_{x^j \in \mX^j} \pi^j(x^j) Q^j(x^{\prime j} \mid x^j,a) \!\!\! & \mbox{if } \text{Den.} = 0
\end{array}
\right.
\end{gather} 




Below we define strategy-belief pair $ (\beta^\star,\mu^\star) $,
\begin{subequations}
	\begin{align}
	\beta^\star = \big( \beta_t^{i,\star} \big)_{t \in \{1,\ldots,T\},i \in \mN}  \qquad  \beta_t^{i,\star}:\mathcal{H}_t^i \rightarrow \Delta(\mA^i) \\
	\mu^\star = \big( \mu_t^{i,\star} \big)_{t \in \{1,\ldots,T\},i \in \mN}  \qquad  \mu_t^{i,\star}:\mathcal{H}_t^i \rightarrow \Delta(\mathcal{H}_t) 
	\end{align}
\end{subequations} 
based on the mapping $ \theta $ produced by the above algorithm. 

Define belief $ \mu^\star $ inductively as follows: set $ \mu_1^\star(x_1) = \prod_{i=1}^N Q^i_0(x_1^i) $. Then for $ t \in \{1,\ldots,T\} $,
\begin{subequations} \label{eqbelfh}
	\begin{align}
	\mu_{t+1}^{i,\star}\big[ h_{t+1}^c \big] 
	&= F\big( \mu_{t}^{i,\star}[h_t^c], \theta_t^i\big[\mu_{t}^{\star}[h_t^c]\big], a_{t} \big)
	\\
	\mu_{t+1}^\star\big[ h_{t+1}^c \big](x_{t+1}) 
	&= \prod_{i=1}^N \mu_{t+1}^{i,\star}\big[ h_{t+1}^c \big](x_{t+1}^i)
	\end{align}
\end{subequations}
Denote the strategy arising out of $ \tilde{\gamma}_t^{i,T} $ by $ \beta^{i,\star}_t $ i.e.,
\begin{gather} \label{eqfhbeta}
\beta_t^{i,\star}(a_t^i \mid h_t^i) = \theta_t^i\big[\mu_t^{\star}[h_t^c]\big](a_t^i \mid x_t^i)
\end{gather}

\subsection{Finite Horizon Result}

In this section, we presented three lemmas. The first two are technical results needed in the proof of the third. The result of the third lemma is used in Section~\ref{secih}.

Define the reward-to-go $ W_t^{i,\beta^i,T} $ for any agent $ i $ and strategy $ \beta^i $  as
\begin{multline} \label{eqr2gfh}
W_t^{i,\beta^i,T}(h_t^i) = \mE^{\beta^i,\beta^{-i,\star},\mu_t^\star[h_t^c]} \big[ \sum_{n=t}^T \delta^{n-t} R^i(X_n,A_n) 
\\
+ \delta^{T+1-t} G^{i}(\underline{\Pi}_{T+1},X^i_{T+1}) \mid h_t^i \big].
\end{multline}
Here agent $ i $'s strategy is $ \beta^i $ whereas all other agents use strategy $ \beta^{-i,\star} $ defined above. Since $ \mX^i,\mA^i $ are assumed to be finite and $ G^i $ absolutely bounded, the reward-to-go is finite $ \forall $ $ i,t,\beta^i,h_t^i $.

\begin{lemma}\label{thmfh1}
	For any $ t \in \{1,\ldots,T\} $, $ i \in \mN $, $ h_t^i $ and $ \beta^i $,	
	\begin{multline} \label{eqintlem}
	V_t^{i,T}(\mu_t^\star[h_t^c],x_t^i) \ge \mE^{\beta^i,\beta^{-i,\star},\mu_t^\star[h_t^c]} \big[ R^i(X_t,A_t) 
	\\
	+ \delta V_{t+1}^{i,T}\big( \big[F(\mu_t^{j,\star}[h_t^c],\beta_t^{j,\star},A_t)\big]_{j=1}^N , X_{t+1}^i \big) \mid h_t^i \big]
	\end{multline}
\end{lemma}
\begin{proof}
	We use proof by contradiction. Suppose $ \exists $ $ i,t,\hat{h}_t^i,\hat{\beta}^i $ such that~\eqref{eqintlem} is violated. Construct strategy $ \hat{\gamma}_t^i $ as  
	\begin{gather} \label{eqnewst}
	\hat{\gamma}_t^i(a_t^i \mid x_t^i) = \left\{
	\begin{array}{ll}
	\hat{\beta}_t^i(a_t^i \mid \hat{h}_t^i)   & \mbox{if } x_t^i = \hat{x}_t^i 
	\\
	\frac{1}{\vert \mA^i \vert} & \mbox{if } x_t^i \ne \hat{x}_t^i
	\end{array}
	\right.
	\end{gather}	
	Then
	\begin{subequations}
		\begin{align}
		&V_t^{i,T}(\mu_t^\star[\hat{h}_t^c],\hat{x}_t^i) 
		\\
		&\ge \mE^{\hat{\gamma}_t^i(\cdot \mid \hat{x}_t^i),\beta_t^{-i,\star},\mu_t^\star[\hat{h}_t^c]} \big[ R^i(X_t,A_t)  \delta V_{t+1}^{i,T}
		\\ \nonumber
		&+ \big( \big[ F(\mu_t^{j,\star}[\hat{h}_t^c],\beta_t^{j,\star}(\cdot \mid \hat{h}_t^c ,\cdot),A_t) \big]_{j=1}^N, X_{t+1}^{i} \big) \mid \mu_t^\star[\hat{h}_t^c],\hat{x}_t^i \big]
		\\
		&= \mE^{\hat{\beta}_t^i,\beta_t^{-i,\star},\mu_t^\star[\hat{h}_t^c]} \big[ R^i(X_t,A_t)  + \delta V_{t+1}^{i,T}
		\\ \nonumber 
		&\big( \big[ F(\mu_t^{j,\star}[\hat{h}_t^c],\beta_t^{j,\star}(\cdot \mid \hat{h}_t^c ,\cdot),A_t) \big]_{j=1}^N, X_{t+1}^{i} \big) \mid \hat{h}_t^i \big]
		\\
		&> V_t^{i,T}(\mu_t^\star[\hat{h}_t^c],\hat{x}_t^i)
		\end{align}
	\end{subequations}
	The first inequality above follows from the algorithm definition in~\eqref{eqfhalg} and~\eqref{eqfhalg2}, the second equality follows from the definition above in~\eqref{eqnewst} and finally the last inequality follows from the assumption at the beginning of this proof. 
	
	Since the above is clearly a contradiction, the result follows.
\end{proof}

\begin{lemma} \label{lemcond}
	\begin{multline}
	\mE^{\beta^i_{t+1:T},\beta_{t+1:T}^{-i,\star},\mu_{t+1}^\star[h_t^c,a_t]} \big[ \sum_{n=t+1}^T \delta^{n-(t+1)} R^i(X_n,A_n) 
	\\
	+ \delta^{T+1-t} G^i(\Pi_{T+1},X_{T+1}^i) \mid h_t^i,a_t,x_{t+1}^i \big]
	\\
	= \mE^{\beta^i_{t:T},\beta_{t:T}^{-i,\star},\mu_{t}^\star[h_t^c]} \big[ \sum_{n=t+1}^T \delta^{n-(t+1)} R^i(X_n,A_n) 
	\\
	+ \delta^{T+1-t} G^i(\Pi_{T+1},X_{T+1}^i) \mid h_t^i,a_t,x_{t+1}^i \big]
	\end{multline}
\end{lemma}

\begin{proof}
	This result relies on the structure of the update in~\eqref{eqbelfh}, specifically that $ \mu_{t+1}^{-i,\star}[h_{t+1}^c] $ is a deterministic function of $ \mu_t^{-i,\star}[h_t^c] $, $ \beta_t^{-i,\star} $, $ a_t $ and does not depend on $ \beta^i $.  
	
	Consider the  joint pmf-pdf of random variables involved in the expectation
	\begin{multline} \label{eqnumden2}
	\mP^{\beta_{t:T}^i,\beta_{t:T}^{-i,\star},\mu_t^\star[h_t^c] }\big( x_{t+1}^{-i},a_{t+1:T}, x_{t+2:T},x_{T+1}^i,
	\\
	\pi_{T+1} \mid h_t^i,a_t,x_{t+1}^i \big) 
	\end{multline}
	This can be written as $ \frac{A}{B} $ with
	\begin{subequations} \label{eqnumden}
		\begin{multline}
		A = \sum_{x_t^{-i}} \mP^{\beta_{t:T}^i,\beta_{t:T}^{-i,\star},\mu_t^\star[h_t^c]} \big( x_t^{-i},a_t,x_{t+1},a_{t+1:T},x_{t+2:T},
		\\
		x_{T+1}^i,\pi_{T+1} \mid h_t^i \big) 
		\end{multline} 
		\begin{gather} 
		B = \sum_{\tilde{x}_t^{-i}} \mP^{\beta_{t:T}^i,\beta_{t:T}^{-i,\star},\mu_t^\star[h_t^c]}\big( \tilde{x}_t^{-i}, a_t, x_{t+1}^i \mid h_t^i \big)  
		\end{gather}
	\end{subequations} 
	
	Using causal decomposition we can write
	\begin{subequations}
		\begin{multline}
		A = \sum_{x_t^{-i}} \mP^{\beta_{t:T}^i,\beta_{t:T}^{-i,\star},\mu_t^\star[h_t^c]}(x_t^{-i} \mid h_t^{i}) \beta_t^i(a_t^i \mid h_t^i) 
		\\
		\beta_t^{-i,\star}(a_t^{-i} \mid h_t^c,x_t^{-i}) Q(x_{t+1} \mid x_t,a_t) 
		\\
		\mP^{\beta_{t:T}^i,\beta_{t:T}^{-i,\star},\mu_t^\star[h_t^c]}( a_{t+1:T},x_{t+2:T},x_{T+1}^i,\pi_{T+1} 
		\\
		\mid h_t^i,a_t,x_t^{-i},x_{t+1} )
		\end{multline}
		\begin{multline} 
		= \sum_{x_t^{-i}} \mu_t^{-i,\star}[h_t^c](x_t^{-i}) \beta_t^i(a_t^i \mid h_t^i) \beta_t^{-i,\star}(a_t^{-i} \mid h_t^c,x_t^{-i}) 
		\\
		Q^i(x_{t+1}^i \mid x_t^i,a_t) Q^{-i}(x_{t+1}^{-i} \mid x_t^{-i},a_t)
		\\
		\mP^{\beta_{t+1:T}^i,\beta_{t+1:T}^{-i,\star},\mu_{t+1}^\star[h_{t+1}^c]}( a_{t+1:T},x_{t+2:T},x_{T+1}^i,\pi_{T+1} 
		\\
		\mid h_t^i,a_t,x_t^{-i},x_{t+1} )
		\end{multline} 
	\end{subequations}
	where the second equality follows from the fact that given $ h_t^i,a_t,x_t^{-i},x_{t+1} $ and $ \mu_t^\star[h_t^c] $, the probability of $ (a_{t+1:T},x_{t+2:T},x_{T+1}^i,\pi_{T+1} ) $ depends on $ h_t^i,a_t,x_t^{-i},x_{t+1},\mu_{t+1}^\star[h_{t+1}^c] $ only through $ \beta_{t+1:T}^i,\beta_{t+1:T}^{-i,\star} $. Also the second equality above uses the fact that types evolve conditionally independent given action. Performing similar decomposition of the denominator $ B $ and substituting back in the expression from~\eqref{eqnumden2} allows us to cancel the terms $ \beta^i(\cdot) $ and $ Q^i(\cdot) $. Using the belief update from~\eqref{eqbelfh}, this gives that the expression in~\eqref{eqnumden2} is 
	\begin{subequations}
		\begin{multline}
		\mu_{t+1}^{-i,\star}[h_{t+1}^c](x_{t+1}^{-i}) \mP^{\beta_{t+1:T}^i,\beta_{t+1:T}^{-i,\star},\mu_{t+1}^\star[h_{t+1}^c]}( a_{t+1:T},x_{t+2:T},
		\\
		x_{T+1}^i,\pi_{T+1} \mid h_t^i,a_t,x_{t+1} )
		\\
		= \mP^{\beta_{t+1:T}^i,\beta_{t+1:T}^{-i,\star},\mu_{t+1}^\star[h_{t+1}^c]}( x_{t+1}^{-i}, a_{t+1:T},x_{t+2:T},
		\\
		x_{T+1}^i,\pi_{T+1} \mid h_t^i,a_t,x_{t+1}^i )
		\end{multline} 	
	\end{subequations}  
	The above equality follows directly from definition.
	
	This completes the proof.
\end{proof}

The result below shows that the value function from the backwards recursive algorithm is higher than any reward-to-go.

\begin{lemma}\label{thmfh2}
	For any $ t \in \{1,\ldots,T\} $, $ i \in \mN $, $ h_t^i $ and $ \beta^i $,	
		\begin{gather}
		V_t^{i,T}(\mu_t^\star[h_t^c],x_t^i) \ge W_t^{i,\beta^i,T}(h_t^i)
		\end{gather}
\end{lemma}

\begin{proof}
	We use backward induction for this. At time $ T $, using the maximization property from~\eqref{eqfhalg}, 
	\begin{subequations}
		\begin{align}
		&V_T^{i,T}(\mu_T^\star[h_T^c],x_T^i) 
		\\
		&\triangleq \mE^{\tilde{\gamma}_T^{i,T}(\cdot \mid x_T^i),\tilde{\gamma}_T^{-i,T},\mu_T^\star[h_t^c]} \big[ R^i(X_T,A_T) 
		\\ \nonumber 
		&+ \delta G^i\big( \big[F(\mu_T^{j,\star}[h_T^c],\tilde{\gamma}_T^{j,T},A_T) \big]_{j=1}^N ,X_{T+1}^i \big) \mid \mu_T^\star[h_T^c],x_T^i \big]
		\\
		&\ge \mE^{{\gamma}_T^{i,T}(\cdot \mid x_T^i),\tilde{\gamma}_T^{-i,T},\mu_T^\star[h_t^c]} \big[ R^i(X_T,A_T) 
		\\ \nonumber 
		&+ \delta G^i\big( \big[F(\mu_T^{j,\star}[h_T^c],\tilde{\gamma}_T^{j,T},A_T) \big]_{j=1}^N ,X_{T+1}^i \big) \mid \mu_T^\star[h_T^c],x_T^i \big]
		\\
		&= W_T^{i,\beta^i,T}(h_T^i)
		\end{align}
	\end{subequations}
Here the second inequality follows from~\eqref{eqfhalg} and~\eqref{eqfhalg2} and the final equality is by definition in~\eqref{eqr2gfh}.
	
	Assume that the result holds for all $ n \in \{t+1,\ldots,T\} $, then at time $ t $ we have
	\begin{subequations}
		\begin{align}
		&V_t^{i,T}(\mu_t^\star[h_t^c],x_t^i) 
		\\
		&\ge \mE^{\beta_t^i,\beta_t^{-i,\star},\mu_t^\star[h_t^c]} \big[ R^i(X_t,A_t) 
		\\ \nonumber 
		&+ \delta V_{t+1}^{i,T}\big( \big[F(\mu_t^{j,\star}[h_t^c],\beta_t^{j,\star},A_t)\big]_{j=1}^N , X_{t+1}^i \big) \mid h_t^i \big]
		\\
		&\ge \mE^{\beta_t^i,\beta_t^{-i,\star},\mu_t^\star[h_t^c]} \big[ R^i(X_t,A_t) 
		\\ \nonumber 
		&+ \delta \mE^{\beta^i_{t+1:T},\beta_{t+1:T}^{-i,\star},\mu_{t+1}^\star[h_t^c,A_t]} \big[ \sum_{n=t+1}^T \delta^{n-(t+1)} R^i(X_n,A_n) 
		\\ \nonumber 
		&+ \delta^{T-t} G^i(\Pi_{T+1},X_{T+1}^i) \mid h_t^i,A_t,X_{t+1}^i \big] \mid h_t^i \big]
		\\
		&= \mE^{\beta^i_{t:T},\beta^{-i,\star}_{t:T},\mu_t^\star[h_t^c]} \big[ \sum_{n=t}^T \delta^{n-t} R^i(X_n,A_n) 
		\\ \nonumber 
		&+ \delta^{T+1-t}G^i(\Pi_{T+1},X_{T+1}^i) \mid h_t^i \big]
		\\
		&= W_t^{i,\beta^i,T}(h_t^i)
		\end{align}
	\end{subequations}
	Here the first inequality follows from Lemma~\ref{thmfh1}, the second inequality from the induction hypothesis, the third equality follows from Lemma~\ref{lemcond} and the final equality by definition~\eqref{eqr2gfh}.
\end{proof}

\section{SPBE in Infinite Horizon} \label{secih}

In this section we consider the infinite horizon dynamic game, with naturally no terminal reward.



\subsection{Perfect Bayesian Equilibrium}

A perfect Bayesian equilibrium is the pair of strategy and belief $ (\beta^\star,\mu^\star) $, where
\begin{subequations}
\begin{align}
\beta^\star = \big( \beta_t^{i,\star} \big)_{t \ge 1,i \in \mN}  \qquad  \beta_t^{i,\star}:\mathcal{H}_t^i \rightarrow \Delta(\mA^i) \\
\mu^\star = \big( \mu_t^{i,\star} \big)_{t \ge 1,i \in \mN}  \qquad  \mu_t^{i,\star}:\mathcal{H}_t^i \rightarrow \Delta(\mathcal{H}_t) 
\end{align}
\end{subequations} 
such that sequential rationality is satisfied: $ \forall $ $ i \in \mN $, $ \beta^i $, $ t \ge 1 $, $ h_t^i \in \mathcal{H}_t^i $, 
\begin{multline}
\mE^{\beta^{i,\star},\beta^{-i,\star},\mu_t^\star[h_t^c]} \big[ \sum_{n=t}^\infty \delta^{n-t} R^i(X_n,A_n) \mid h_t^i \big] 
\\
\ge
\mE^{\beta^{i},\beta^{-i,\star},\mu_t^\star[h_t^c]} \big[ \sum_{n=t}^\infty \delta^{n-t} R^i(X_n,A_n) \mid h_t^i \big]
\end{multline}
and beliefs satisfy certain consistency conditions (please refer~\cite[pp.~331]{fudenbergtirole} for the exact conditions).

\subsection{Fixed Point Equation for Infinite Horizon}

In this section, we state the fixed-point equation that defines the value function and strategy mapping for the infinite horizon problem. This is analogous to the backwards recursion (\eqref{eqfhalg} and~\eqref{eqfhalg2}) that defined the value function and $ \theta $ mapping for the finite horizon problem.

Define the set of functions $ V^i: \times_{j=1}^N \Delta(\mX^j) \times \mX^i \rightarrow \mathbb{R} $ and strategies $ \tilde{\gamma}^i:\mX^i \rightarrow \Delta(\mA^i) $ (which is generated formally as $ \tilde{\gamma}^i = \theta^i[\upi] $ for given $ \upi $) via the following fixed-point equation: $ \forall $ $ i \in \mN $, $ x^i \in \mX^i $,
\begin{subequations} \label{eqihfpe}
\begin{multline}
\tilde{\gamma}^i(\cdot \mid x^i) \in \argmax_{\gamma^i(\cdot \mid x^i) \in \Delta(\mA^i)} \mE^{\gamma^i(\cdot \mid x^i),\tilde{\gamma}^{-i},\pi^{-i}} \big[ R^i(X,A) 
\\
+ \delta V^i\big([F(\pi^j,\tilde{\gamma}^j,A^j)]_{j=1}^N,X^{i,\prime}\big) \mid \upi,x^i \big],
\end{multline}
\begin{multline}
V^i(\upi,x^i) = \mE^{\tilde{\gamma}^i(\cdot \mid x^i),\tilde{\gamma}^{-i},\pi^{-i}} \big[ R^i(X,A) 
\\
+ \delta V^i\big([F(\pi^j,\tilde{\gamma}^j,A^j)]_{j=1}^N,X^{i,\prime}\big) \mid \upi,x^i \big].
\end{multline}
\end{subequations} 
Note that the above is a joint fixed-point equation in $ (V,\tilde{\gamma}) $, unlike the backwards recursive algorithm earlier which required solving a fixed-point equation only in $ \tilde{\gamma} $. Here the unknown quantity is distributed as
\begin{multline}
(X^{-i},A^i,A^{-i},X^{i,\prime}) \sim \pi^{-i}(x^{-i})\gamma^i(a^i \mid x^i)
\\
\tilde{\gamma}^{-i}(a^{-i} \mid x^{-i}) Q(x^{i,\prime} \mid x^i,a).
\end{multline}
and $ F(\cdot) $ is as defined in~\eqref{eqf}. 

Define belief $ \mu^\star $ inductively as follows: set $ \mu_1^\star(x_1) = \prod_{i=1}^N Q^i_0(x_1^i) $. Then for $ t \ge 1 $,
\begin{subequations} \label{eqmuih}
\begin{align}
\mu_{t+1}^{i,\star}\big[ h_{t+1}^c \big] 
&= F\big( \mu_{t}^{i,\star}[h_t^c], \theta^i\big[\mu_{t}^{\star}[h_t^c]\big], a_{t} \big)
\\
\mu_{t+1}^\star\big[ h_{t+1}^c \big](x_{t+1}) 
&= \prod_{i=1}^N \mu_{t+1}^{i,\star}\big[ h_{t+1}^c \big](x_{t+1}^i)
\end{align}
\end{subequations}
By construction the belief defined above satisfies the consistency condition needed for a Perfect Bayesian Equilibrium. Denote the stationary strategy arising out of $ \tilde{\gamma} $ by $ \beta^\star $ i.e.,
\begin{gather} \label{eqbeta}
\beta_t^{i,\star}(a_t^i \mid h_t^i) = \theta^i\big[\mu_t^{\star}[h_t^c]\big](a_t^i \mid x_t^i)
\end{gather}

\subsection{Relation between Infinite and Finite Horizon problems} 

The following result highlights the similarities between the fixed-point equation in infinite horizon and the backwards recursion in the finite horizon. 

\begin{lemma}\label{lemfh}
Consider the finite horizon game with $ G^i \equiv V^i $. Then $ V_t^{i,T} = V^i $,  $ \forall $ $ i \in \mN $, $ t \in \{1,\ldots,T\} $ satisfies the backwards recursive construction~\eqref{eqfhalg} and~\eqref{eqfhalg2}.
\end{lemma}	
\begin{proof}
	Use backward induction for this. Consider the finite horizon algorithm at time $ t=T $, noting that $ V_{T+1}^{i,T} \equiv G^i \equiv  V^i $, 
	\begin{subequations} \label{eqfhT}
		\begin{multline}
		\tilde{\gamma}_T^{i,T}(\cdot \mid x_T^i) \in\!\!\!\!\! \argmax_{\gamma_T^i(\cdot \mid x_T^i) \in \Delta(\mA^i)} \!\!\! \mE^{\gamma_T^i(\cdot \mid x_T^i),\tilde{\gamma}_T^{-i,T},\upi_T^{-i}} \big[ R^i(X_T,A_T) 
		\\
		+ \delta V^i\big( \big[F(\pi_T^j, \tilde{\gamma}_T^{j,T}, A_T)\big]_{j=1}^N , X_{T+1}^i \big) \mid \upi_T,x_T^i \big]
		\end{multline}
		\begin{multline}
		V_T^{i,T}(\upi_T,x_T^i) = \mE^{\tilde{\gamma}_T^{i,T}(\cdot \mid x_T^i),\tilde{\gamma}_T^{-i,T},\upi_T^{-i}} \big[ R^i(X_T,A_T)
		\\ 
		+ \delta V^i\big( \big[F(\pi_T^j, \tilde{\gamma}_T^{j,T}, A_T)\big]_{j=1}^N , X_{T+1}^i \big) \mid \upi_T,x_T^i \big]
		\end{multline}
	\end{subequations}
	Comparing the above set of equations with~\eqref{eqihfpe}, we can see that the pair $ (V,\tilde{\gamma}) $ arising out of~\eqref{eqihfpe} satisfies the above. Now assume that $ V_n^{i,T} \equiv V^i $ for all $ n \in \{t+1,\ldots,T\} $. At time $ t $, in the finite horizon construction from~\eqref{eqfhalg},~\eqref{eqfhalg2}, substituting $ V^i $ in place of $ V_{t+1}^{i,T} $ from the induction hypothesis, we get the same set of equations as~\eqref{eqfhT}. Thus $ V_t^{i,T} \equiv V^i $ satisfies it.
\end{proof}

\subsection{Equilibrium Result}

Below we state the central result of this paper. It states that the strategy-belief pair $ (\beta^\star,\mu^\star) $ constructed from the solution of the fixed-point equation~\eqref{eqihfpe} and the forward recursion of~\eqref{eqmuih} and~\eqref{eqbeta} indeed constitutes a PBE. 

\begin{theorem}\label{thih}
	Assuming that the fixed-point equation~\eqref{eqihfpe} admits an absolutely bounded solution $ V^i $ (for all $ i \in \mN $), the strategy-belief pair $ (\beta^\star,\mu^\star) $ defined in~\eqref{eqmuih} and~\eqref{eqbeta} is a PBE of the infinite horizon discounted reward dynamic game i.e., $ \forall $ $ i \in \mN $, $ \beta^i $, $ t \ge 1 $, $ h_t^i \in \mathcal{H}_t^i $,
	\begin{multline}
	\mE^{\beta^{i,\star},\beta^{-i,\star},\mu_t^\star[h_t^c]} \big[ \sum_{n=t}^\infty \delta^{n-t} R^i(X_n,A_n) \mid h_t^i \big] 
	\\
	\ge
	\mE^{\beta^{i},\beta^{-i,\star},\mu_t^\star[h_t^c]} \big[ \sum_{n=t}^\infty \delta^{n-t} R^i(X_n,A_n) \mid h_t^i \big]
	\end{multline}
\end{theorem}

\paragraph*{Remark} Note that by definition in~\eqref{eqbelfh}, $ \mu^\star $ already satisfies the consistency conditions required for perfect Bayesian equilibrium.

\begin{proof}
We divide the proof into two parts: first we show that the value function $ V^i $ is at least as big as any reward-to-go function; secondly we show that under the strategy $ \beta_i^\star $, reward-to-go is $ V^i $.	
	
\paragraph*{Part 1}
For any $ i \in \mN $, $ \beta^i $ define the following reward-to-go functions
\begin{subequations} \label{eqihr2g}
\begin{gather}
W_t^{i,\beta^i}(h_t^i) = \mE^{\beta^i,\beta^{-i,\star},\mu_t^\star[h_t^c]} \big[ \sum_{n=t}^\infty \delta^{n-t} R^i(X_n,A_n) \mid h_t^i \big] 
\end{gather}
\begin{multline}
W_t^{i,\beta^i,T}(h_t^i) = \mE^{\beta^i,\beta^{-i,\star},\mu_t^\star[h_t^c]} \big[ \sum_{n=t}^T \delta^{n-t} R^i(X_n,A_n) 
\\
+ \delta^{T+1-t} V^{i}(\underline{\Pi}_{T+1},X^i_{T+1}) \mid h_t^i \big]
\end{multline}
\end{subequations} 
Since $ \mX^i,\mA^i $ are finite sets the reward $ R^i $ is absolutely bounded, the reward-to-go $ W_t^{i,\beta^i}(h_t^i) $ is finite $ \forall $ $ i,t,\beta^i,h_t^i $. 
	
For any $ i \in \mN $, $ h_t^i \in \mathcal{H}_t^i $,
\begin{multline}
\label{eqdc} 
V^i\big(\mu_t^\star[h_t^c],x_t^i\big) - W_t^{i,\beta^i}(h_t^i) 
= 
\\
\Big( V^i\big(\mu_t^\star[h_t^c],x_t^i\big) - W_t^{i,\beta^i,T}(h_t^i) \Big) 
+ \Big( W_t^{i,\beta^i,T}(h_t^i) - W_t^{i,\beta^i}(h_t^i) \Big)
\end{multline}
Combining results from Lemma~\ref{thmfh2} and~\ref{lemfh}, the term in the first bracket in RHS of~\eqref{eqdc} is non-negative. Using~\eqref{eqihr2g}, the term in the second bracket is 
\begin{multline} \label{eqdiff}
\left( \delta^{T+1-t} \right) \mE^{\beta^i,\beta^{-i,\star},\mu_t^\star[h_t^c]} \big[- \sum_{n=T+1}^\infty \delta^{n-(T+1)} R^i(X_n,A_n) 
\\
+  V^{i}(\underline{\Pi}_{T+1},X^i_{T+1}) \mid h_t^i \big].
\end{multline}
The summation in the expression above is bounded by a convergent Geometric series. Also, $ V^i $ is bounded. Hence the above quantity can be made arbitrarily small by choosing $ T $ appropriately large. Since the LHS of~\eqref{eqdc} does not depend on $ T $, this gives that 
\begin{gather}
V^i\big(\mu_t^\star[h_t^c],x_t^i\big) \ge W_t^{i,\beta^i}(h_t^i) 
\end{gather}

\paragraph*{Part 2}
Since the strategy $ \beta^\star $ generated in~\eqref{eqbeta} is such that $\beta^{i,\star}_t $ depends on $ h_t^i $ only through $ \mu_t^\star[h_t^c] $ and $ x_t^i $, the reward-to-go $ W_t^{i,\beta^{i,\star}} $, at strategy $ \beta^\star $, can be written (with abuse of notation) as
\begin{multline}
W_t^{i,\beta^{i,\star}}(h_t^i) = W_t^{i,\beta^{i,\star}}(\mu_t^\star[h_t^c],x_t^i) 
\\
= \mE^{\beta^{\star},\mu_t^\star[h_t^c]} \big[ \sum_{n=t}^\infty \delta^{n-t} R^i(X_n,A_n) \mid \mu_t^\star[h_t^c],x_t^i \big]
\end{multline}

For any $ h_t^i \in \mathcal{H}_t^i $,
\begin{subequations}
\begin{multline}
W_t^{i,\beta^{i,\star}}(\mu_t^\star[h_t^c],x_t^i) 
= \mE^{\beta^{\star},\mu_t^\star[h_t^c]} \big[ R^i(X_t,A_t) + \delta W_{t+1}^{i,\beta^{i,\star}}
\\
\big(\big[F(\mu_t^{j,\star}[h_t^c],\theta^i[\mu_t^\star[h_t^c]],A_{t+1}^j)\big]_{j=1}^N,X_{t+1}^i\big)  \mid \mu_t^\star[h_t^c],x_t^i \big]
\end{multline}
\begin{multline}
V^{i}(\mu_t^\star[h_t^c],x_t^i) 
= \mE^{\beta^{\star},\mu_t^\star[h_t^c]} \big[ R^i(X_t,A_t) + \delta V^{i}
\\
\big(\big[F(\mu_t^{j,\star}[h_t^c],\theta^i[\mu_t^\star[h_t^c]],A_{t+1}^j)\big]_{j=1}^N,X_{t+1}^i\big)  \mid \mu_t^\star[h_t^c],x_t^i \big]
\end{multline}
\end{subequations}
Repeated application of the above for the first $ n $ time periods gives
\begin{subequations}
\begin{multline}
W_t^{i,\beta^{i,\star}}(\mu_t^\star[h_t^c],x_t^i) 
= \mE^{\beta^{\star},\mu_t^\star[h_t^c]} \big[ \sum_{m=t}^{t+n-1} \delta^{m-t} R^i(X_t,A_t)  
\\
+ \delta^{n}  W_{t+n}^{i,\beta^{i,\star}}\big(\Pi_{t+n},X_{t+n}^i\big)  \mid \mu_t^\star[h_t^c],x_t^i \big]
\end{multline}
\begin{multline}
V^{i}(\mu_t^\star[h_t^c],x_t^i) 
= \mE^{\beta^{\star},\mu_t^\star[h_t^c]} \big[ \sum_{m=t}^{t+n-1} \delta^{m-t} R^i(X_t,A_t)  
\\
+ \delta^{n}  V^{i}\big(\Pi_{t+n},X_{t+n}^i\big)  \mid \mu_t^\star[h_t^c],x_t^i \big]
\end{multline}
\end{subequations}
Here $ \Pi_{t+n} $ is the $ n-$step belief update under strategy and belief prescribed by $\beta^\star,\mu^\star$. 

Taking difference gives
\begin{subequations}
\begin{multline}
W_t^{i,\beta^{i,\star}}(\mu_t^\star[h_t^c],x_t^i)  - V^{i}(\mu_t^\star[h_t^c],x_t^i) 
\\
= \delta^n \mE^{\beta^{\star},\mu_t^\star[h_t^c]} \big[ W_{t+n}^{i,\beta^{i,\star}}\big(\Pi_{t+n},X_{t+n}^i\big) 
\\
- V^{i}\big(\Pi_{t+n},X_{t+n}^i\big) \mid \mu_t^\star[h_t^c],x_t^i \big]
\end{multline}
\end{subequations}
Taking absolute value of both sides then using Jensen's inequality for $ f(x) = \vert x \vert $ and finally taking supremum over $ h_t^i $ gives us
\begin{multline}
\sup_{h_t^i} \big\vert W_t^{i,\beta^{i,\star}}(\mu_t^\star[h_t^c],x_t^i)  - V^{i}(\mu_t^\star[h_t^c],x_t^i) \big\vert  
\\
\le 
\delta^n \sup_{h_t^i}  \mE^{\beta^{\star},\mu_t^\star[h_t^c]} \big[\big\vert W_{t+n}^{i,\beta^{i,\star}}(\Pi_{t+n},X_{t+n}^i)  
\\
- V^{i}(\mu_t^\star[h_t^c],x_t^i) \big\vert  \mid \mu_t^\star[h_t^c],x_t^i \big]
\end{multline}
Now using the fact that $ W_{t+n},V^i $ are bounded and that we can choose $ n $ arbitrarily large, we get $ \sup_{h_t^i} \vert W_t^{i,\beta^{i,\star}}(\mu_t^\star[h_t^c],x_t^i)  - V^{i}(\mu_t^\star[h_t^c],x_t^i) \vert = 0 $. 
\end{proof}

\section{A Concrete Example} \label{secexample}

In this section, we  consider an infinite horizon version of the public goods example from~\cite[ch.~8, Example 8.3]{fudenbergtirole}. We solve the corresponding fixed point equation (arising out of~\eqref{eqihfpe})  numerically to calculate the mapping $ \theta $ (which in turn generates the perfect Bayesian equilibrium $ (\beta^\star,\mu^\star) $).

The example consists of two symmetric agents. The type space and action sets are $ \mX^1 = \mX^2 = \{x^H,x^L\} $ and $ \mA^1 = \mA^2 = \{0,1\} $. Each agents' type is static and does not vary with time.

The actions represents whether agents are willing to contribute for a common public good. If at least one agent contributes then both agents receive utility $ 1 $ and the agent(s) that contributed receive cost equal to their type. If no one contributes then both agents receive utility $ 0 $. Thus the reward function is
\begin{gather}
R^i(x,a) = \left\{
\begin{array}{ll}
1-x^i & \mbox{if } a^i = 1 
\\
a^{-i} & \mbox{if } a^i = 0
\end{array}
\right.
\end{gather}
where $ a^{-i} $ represents the action taken by the agent other than $ i $. 

We use the following values $ x^H = 1.2 $, $ x^L = 0.2 $ and consider three values $ \delta = 0,0.5,0.95 $. Since type sets have two elements we can represent the distribution $ \pi_1(\cdot) \in \Delta(\mX^1) $ with only $ \pi_1(x^H) \in [0,1] $, similarly for agent $ 2 $. For any $ \upi = (\pi_1,\pi_2) \in [0,1]^2 $, the mapping $ \theta[\upi] $ produces $ \tilde{\gamma}^i(\cdot \mid x^i) $ for every $ i \in \{1,2\} $ and $ x^i \in \mX^i = \{x^H,x^L\} $. 
%
Since the action space contains two elements, we can represent the distribution $ \tilde{\gamma}^i(\cdot \mid x^i) $ by $ \tilde{\gamma}^i(a^i = 1 \mid x^i) $ i.e., the probability of taking action $ 1 $. We solve the fixed-point equation by discretizing the $ \upi- $space $ [0,1]^2 $ and all solutions that we find are symmetric w.r.t. agents i.e., $ \tilde{\gamma}^1(\cdot \mid x^L) $ for $ \upi=(\pi_1,\pi_2) $ is the same as $ \tilde{\gamma}^2(\cdot \mid x^L) $ for $ \upi^\prime = (\pi_2,\pi_1) $ and similarly for type $ x^H $.


For $ \delta = 0 $, the game is instantaneous and for the values considered, we have $ 1 - x^H = -0.2 < 0 $. This implies that whenever agent $ 1 $'s type is $ x^H $, it is instantaneously profitable not to contribute. This gives $ \tilde{\gamma}^1(a^1=1 \mid x^H) = 0 $, for all $ \upi $. Thus we only plot $ \tilde{\gamma}^1(a^1=1 \mid x^L) $; in Fig.~\ref{fig:del0_p1L}. For $ \delta=0 $ the fixed-point equation~\eqref{eqihfpe} is only for the variable $ \tilde{\gamma} $ and not $ V $, and can be solved analytically. Refer to~\cite[eq.~(20) and Fig.~(1)]{vasal2015aa}, where this solution is stated. There are multiple solutions to the fixed-point equation and our result from Fig.~\ref{fig:del0_p1L} matches with the one of the results in~\cite{vasal2015aa}.

Intuitively, with type $ x^L $ the only value of $ \upi $ for which agent 1 would not wish to contribute is if he anticipates agent $ 2 $'s type to be $ x^L $ with high probability and rely on agent 2 to contribute. This is why for lower values of $ \pi_2 $ (i.e., agent $ 2 $'s type likely to be $ x^L $) we see $ \tilde{\gamma}^1(a^1 = 1 \mid x^L) = 0 $ in Fig.~\ref{fig:del0_p1L}.

Now consider $ \tilde{\gamma}^1(a^1=1 \mid x^L) $ plotted in Fig.~\ref{fig:del0_p1L},~\ref{fig:del05_p1L} and~\ref{fig:del095_p1L}. As $ \delta $ increases, future rewards attain more priority and signaling comes into play. So while taking an action, agents not only look for their instantaneous reward but also how their action affects the future public belief $ \pi $ about their private type. It is evident in the figures that as $ \delta $ increases, at high $ \pi_1 $, up to larger values of $ \pi_2 $ agent $ 1 $ chooses not to contribute when his type is $ x^L $. This way he intends to send a ``wrong'' signal to agent $ 2 $ i.e., that his type is $ x^H $ and subsequently force agent $ 2 $ to invest. This way agent 1 can free-ride on agent $ 2 $'s investment.

Now consider Fig.~\ref{fig:del05_p1H} and~\ref{fig:del095_p1H}, where $ \tilde{\gamma}^1(a^1 = 1 \mid x^H) $ is plotted. Coordination via signaling is evident here. Although it is instantaneously not profitable to contribute if agent $ 1 $'s type is $ x^H $, by contributing at higher values of $ \pi_2 $ (i.e., agent $ 2 $'s type is likely $ x^H $) and low $ \pi_1 $, agent $ 1 $ coordinates with agent $ 2 $ to achieve net profit greater than $ 0 $ (reward when no one contributes). This can be done since the loss of contributing is $ -0.2 $ whereas profit from free-riding on agent $ 2 $'s contribution is $ 1 $. 

Under the equilibrium strategy, beliefs $ \underline{\Pi}_t $ form a Markov chain. One can trace this Markov chain to study the signaling effect at equilibrium. On numerically simulating this Markov chain for the above example (at $ \delta= 0.95 $) we observe that for almost all initial beliefs, within a few rounds agents completely learn each other's private type truthfully (or at least with very high probability). In other words, agents manage to reveal their private type via their actions at equilibrium and to such an extent that it negates any possibly incorrect initial belief about their type.

As a measure of cooperative coordination at equilibrium one can perform the following calculation. Compare the value function $ V^1(\cdot,x) $ of agent $ 1 $ arising out of the fixed-point equation, for $ \delta = 0.95 $ and $ x \in \{x^H,x^L\} $  (normalize it by multiplying with $ 1-\delta $ so that it represents per-round value) with the best possible attainable single-round reward under a symmetric mixed strategy with a) full coordination and b) no coordination. Note that the two cases need not be equilibrium themselves, which is why this will result in a bound on the efficiency of the evaluated equilibria.

In case a), assuming both agents have the same type $ x $, full coordination can lead to the best possible reward of $ \frac{1+1-x}{2} = 1- \frac{x}{2} $ i.e., agent $ 1 $ contributes with probability $ 0.5 $ and agent $ 2 $ contributes with probability $ 0.5 $ but in a coordinated manner so that it doesn't overlap with agent $ 1 $ contributing. 

In case b) when agents do not coordinate and invest with probability $ p $ each, then the expected single-round reward is $ p(1-x) + p(1-p) $. The maximum possible value of this expression is $ (1-\frac{x}{2})^2 $.

For $ x=x^L=0.2 $, the range of values of $ V^1(\pi_1,\pi_2,x^L) $ over $ (\pi_1,\pi_2) \in [0,1]^2 $ is $ [0.865,0.894] $. 
Whereas full coordination produces $ 0.9 $ and no coordination $ 0.81 $. It is thus evident that agents at equilibrium end up achieving reward close to the best possible and gain significantly compared to the strategy of no coordination.

Similarly for $ x = x^H = 1.2 $ the range is $ [0.3,0.395] $. Whereas full coordination produces $ 0.4 $ and no coordination $ 0.16 $. The gain via coordination is evident here too. 


\begin{figure}[!ht] 
	\centering
	\includegraphics[width=0.5\textwidth]{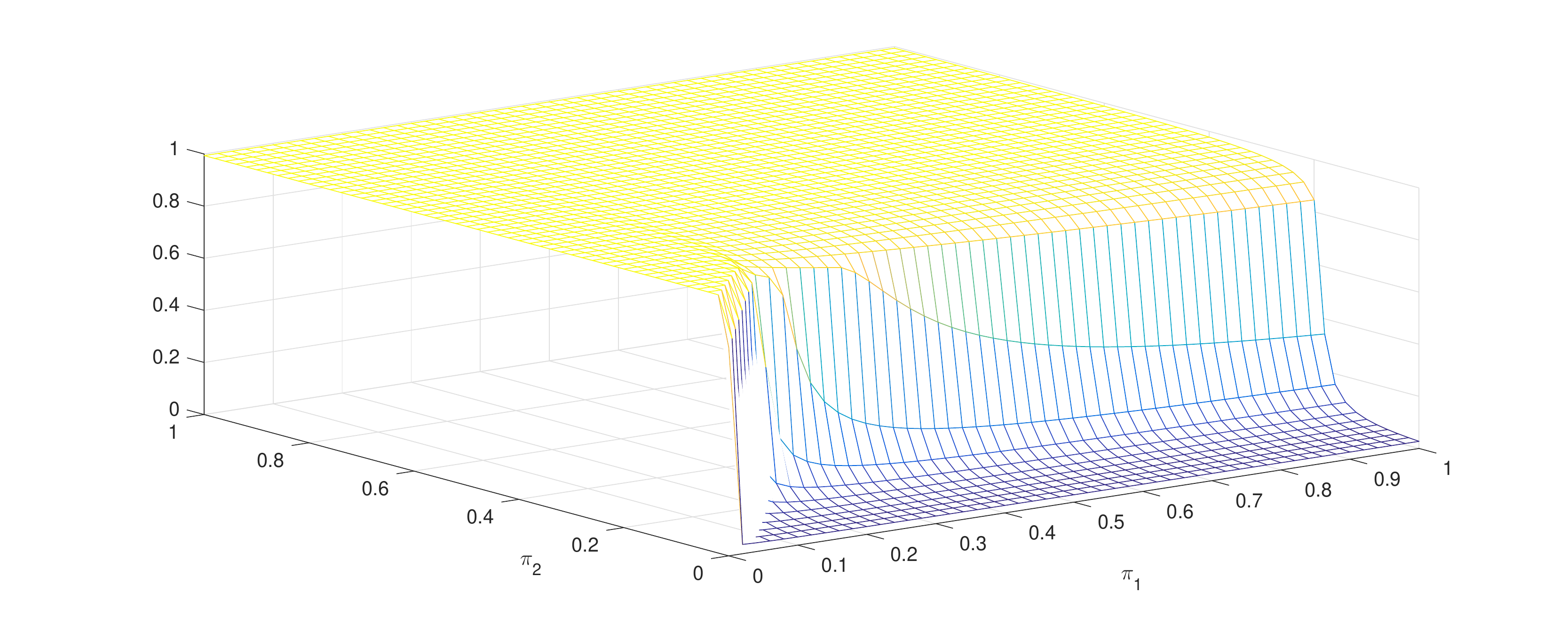}
	\caption{$ \tilde{\gamma}^1(a^1 = 1 \mid x^L) $ vs. $ (\pi_1,\pi_2) $ at $ \delta = 0 $.} 
	\label{fig:del0_p1L}
\end{figure}

\begin{figure}[!ht] 
	\centering
	\includegraphics[width=0.5\textwidth]{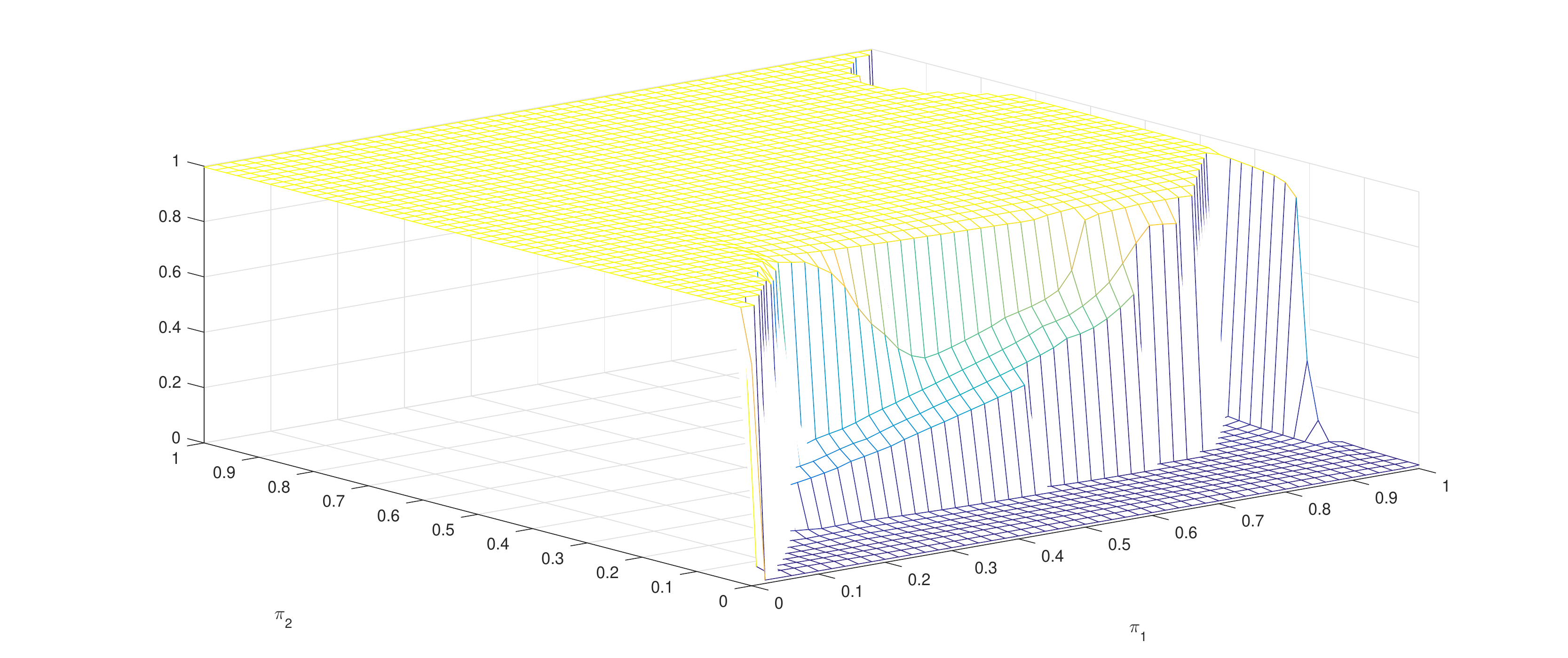}
	\caption{$ \tilde{\gamma}^1(a^1 = 1 \mid x^L) $ vs. $ (\pi_1,\pi_2) $ at $ \delta = 0.5 $.} 
	 \label{fig:del05_p1L}
\end{figure}

\begin{figure}[!ht] 
	\centering
	\includegraphics[width=0.5\textwidth]{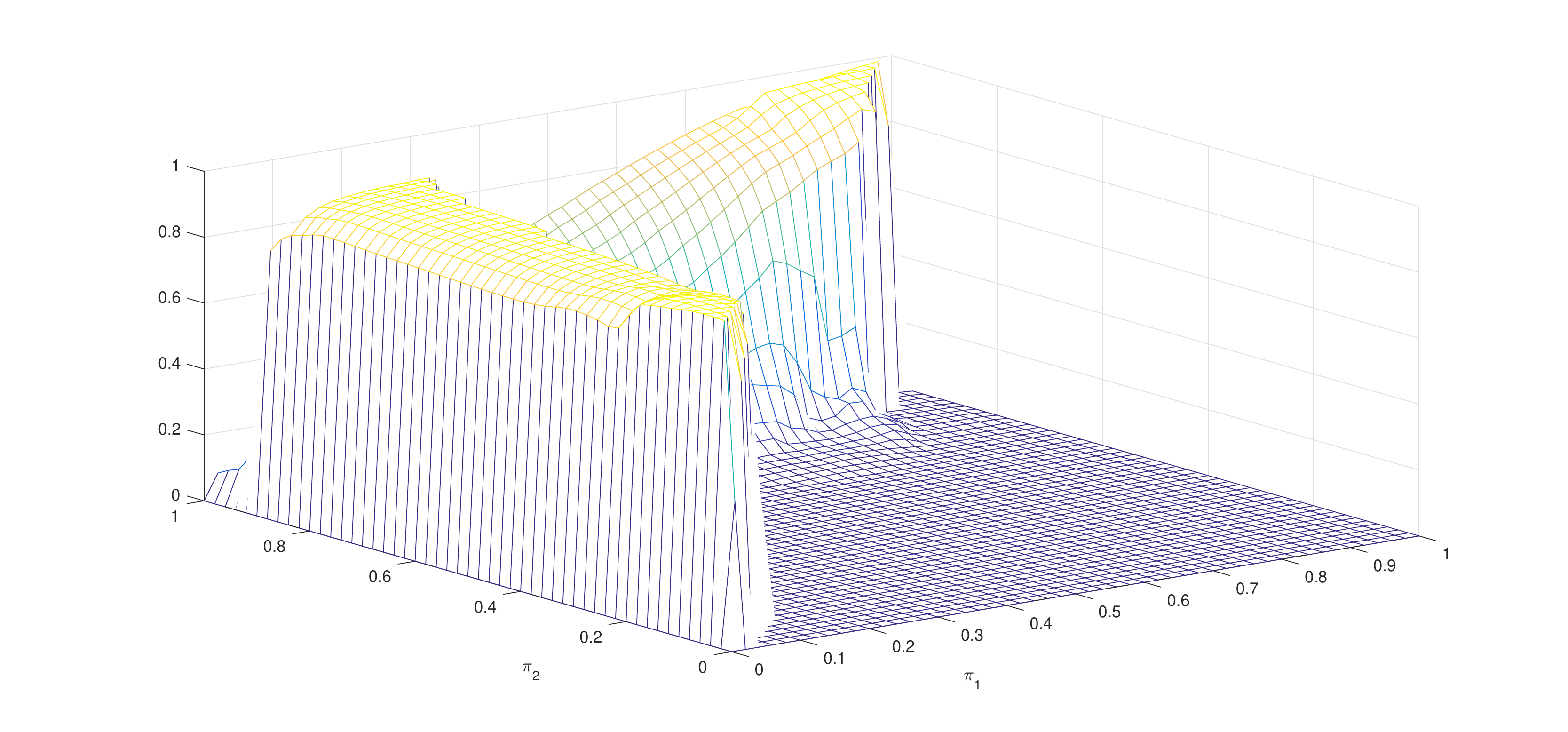}
	\caption{$ \tilde{\gamma}^1(a^1 = 1 \mid x^H) $ vs. $ (\pi_1,\pi_2) $ at $ \delta = 0.5 $.} 
	\label{fig:del05_p1H}
\end{figure}

\begin{figure}[!ht]  
	\centering
	\includegraphics[width=0.5\textwidth]{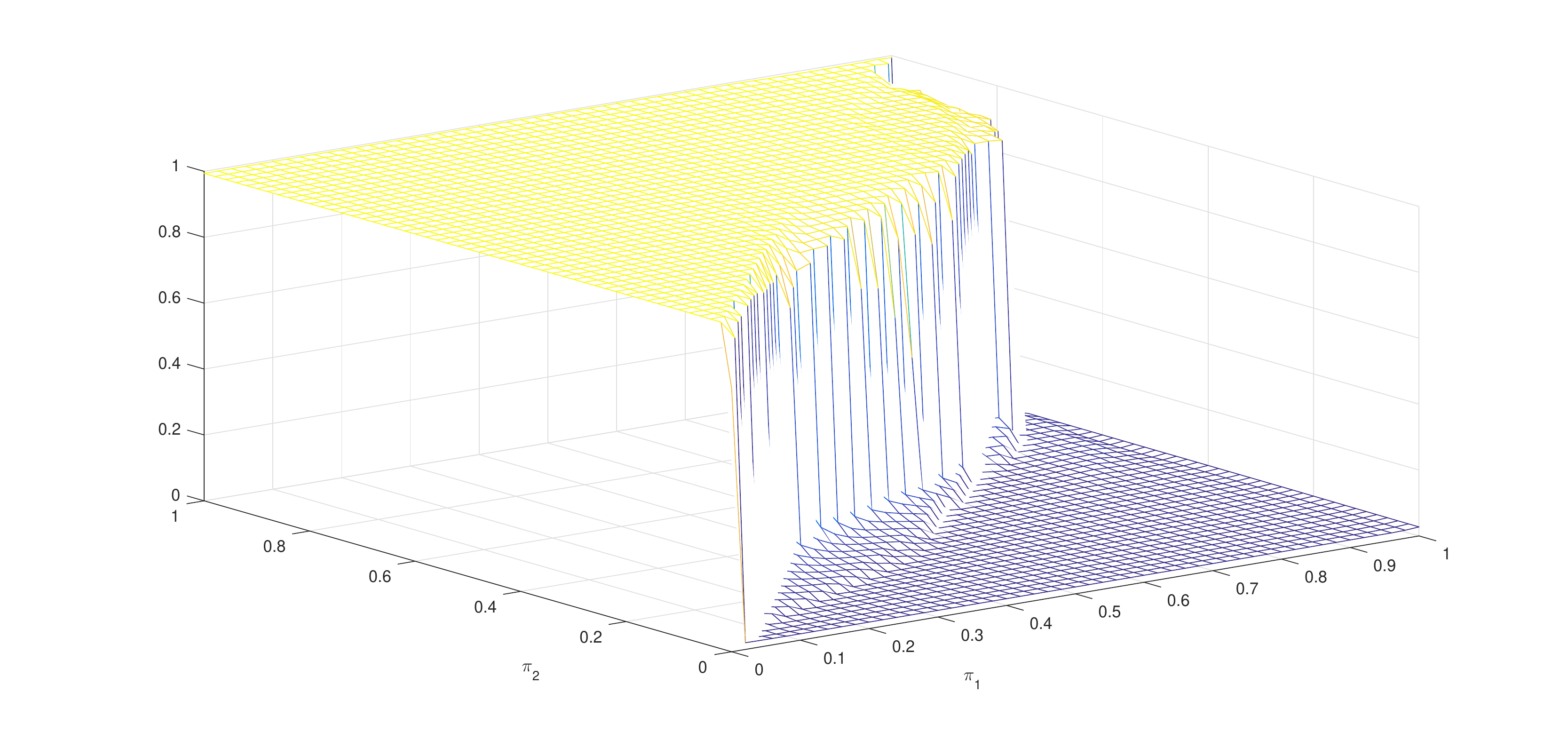}
	\caption{$ \tilde{\gamma}^1(a^1 = 1 \mid x^L) $ vs. $ (\pi_1,\pi_2) $ at $ \delta = 0.95 $.} 
	\label{fig:del095_p1L}
\end{figure}

\begin{figure}[!ht]  
	\centering
	\includegraphics[width=0.5\textwidth]{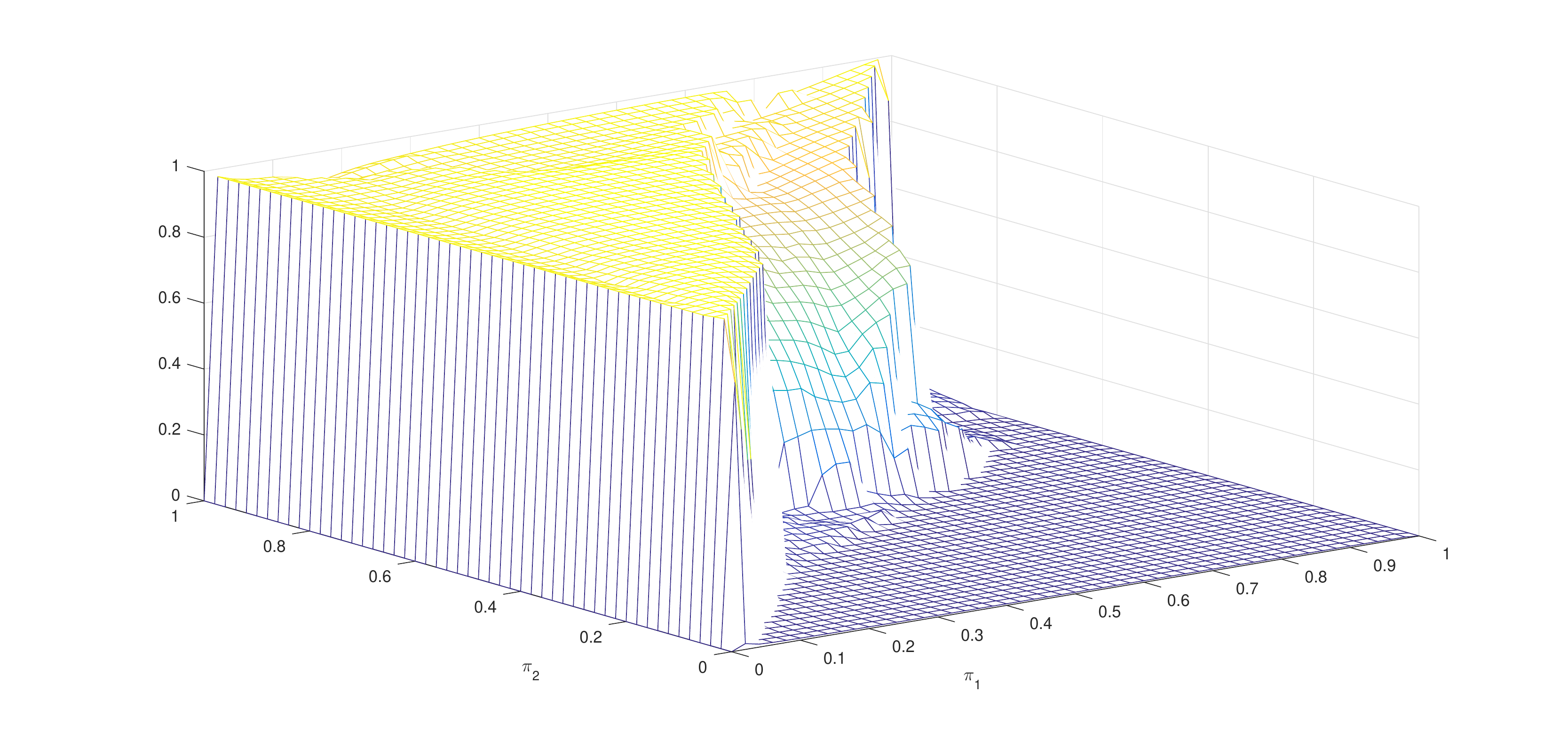}
	\caption{$ \tilde{\gamma}^1(a^1 = 1 \mid x^H) $ vs. $ (\pi_1,\pi_2) $ at $ \delta = 0.95 $.} 
	\label{fig:del095_p1H}
\end{figure}

\section{Conclusion} 

This paper considers the infinite horizon discounted reward dynamic game with private types i.e., where each agent can only observe their own type. The types evolve as a controlled Markov process and are conditionally independent across agents given the action profile. Asymmetry of information between agents exists in this model, since each agent only knows their own private type.

To date, there exists no universal algorithm for calculating PBE in models with asymmetry of information that decouples, w.r.t. time, the calculation of strategy. Section~\ref{secih} provides a single-shot fixed-point equation for calculating the equilibrium generating function $ \theta $, which in conjunction with the forward recursion in~\eqref{eqmuih} and~\eqref{eqbeta} gives a subset of PBEs $ (\beta^\star,\mu^\star) $ of this game. The method proposed in this paper finds PBE of a certain type i.e., where for any agent $ i $, his strategy at equilibrium depends on the public history only through the common belief $ \upi_t $ and on private history only through agent $ i $'s current private type $ x_t^i $.  

Finally, we demonstrate our methodology by a concrete example of a two agent symmetric public goods game and observe the signaling effect in agents' strategies at equilibrium as discount factor is increased. The signaling effect implies that agents take into account how their actions affect future public beliefs $ \upi $ about their private type.

One important direction for future work is characterization of games where the proposed SPBE exists. This boils down to existence of a solution to the fixed-point equation in Section~\ref{secih}, as any SPBE must satisfy this equation.


%

\end{document}